\documentclass[11pt,a4paper,english]{amsart}
\usepackage{babel}
\usepackage[latin1]{inputenc}
\usepackage{amsmath}
\usepackage{amsthm}
\usepackage{amsfonts}
\usepackage{indentfirst}
\usepackage{graphicx}
\usepackage{amssymb}
\usepackage[mathscr]{eucal}

\usepackage{tikz}
\usepackage{caption}

\usepackage{enumerate}
\usepackage{amsthm}
\usepackage{amscd}
\newtheorem{teo}{Theorem}%[section]

\newtheorem{pro}[teo]{Proposition}%[section]
\newtheorem{lem}[teo]{Lemma}%[section]

\theoremstyle{definition}
\newtheorem{rem}[teo]{Remark}%[section]
\newtheorem{de}[teo]{Definition}%[section]
\newtheorem{exa}{Example}%[section]
\title[New binary and ternary LCD codes]{New binary and ternary LCD codes}
\author{Carlos Galindo, Olav Geil, Fernando Hernando and Diego Ruano}
\curraddr{\texttt{Carlos Galindo and Fernando Hernando:} Instituto
Universitario de Matem\'aticas y Aplicaciones de Castell\'on and
Departamento de Matem\'aticas, Universitat Jaume I, Campus de Riu
Sec. 12071 Castell\'{o} (Spain)\\
\texttt{Olav Geil and Diego Ruano:} Department of Mathematical Sciences, Aalborg University, Skjernvej 4A, 9220 Aalborg East (Denmark).
}
\email{{\rm Galindo:} galindo@uji.es; {\rm Geil:} olav@math.aau.dk; {\rm Hernando:} carrillf@uji.es; {\rm Ruano:} diego@math.aau.dk}
\date{}
\thanks{Supported by the Spanish Ministry of Economy/FEDER-UE (grants MTM2015-65764-C2-2-P and MTM2015-69138-REDT), the University Jaume I (grant PB1-1B2015-02) and the Danish Council for Independent Research (grant DFF-4002-00367).}


\begin{document}

\begin{abstract}
LCD codes are linear codes with important cryptographic applications. Recently, a method has been presented to transform any linear code into an LCD code with the same parameters when it is supported on a finite field with cardinality larger than 3. Hence, the study of LCD codes is mainly open for binary and ternary fields. Subfield-subcodes of $J$-affine variety codes are a generalization of BCH codes which have been successfully used for constructing good quantum codes. We describe binary and ternary LCD codes constructed as subfield-subcodes of $J$-affine variety codes and provide some new and good LCD codes coming from this construction.
\end{abstract}

\maketitle

\section*{Introduction}

It is well-known that the hull $C \cap C^\bot$ of a linear code $C$, with (Euclidean) dual $C^\perp$, does not vanish in general; but when this holds, the code $C$ is called a linear code with  complementary dual (LCD).  LCD codes were introduced by Massey \cite{Massey} to provide an optimum linear coding solution for the two-user binary adder channel and prove the existence of asymptotically good LCD codes; previously he studied LCD cyclic codes (reversible codes) in \cite{mass}. The literature contains considerable information about the characterization and construction of this family of codes, being \cite{t-h, ya-ma, sen} some of the oldest references. Apart from applications in data storage, LCD codes are also useful for obtaining lattices \cite{20P} and in network coding \cite{4P, 45P}. Interesting applications of LCD codes in cryptography have been recently discovered. They play a role in counter-measures to passive and active side-channel analyses on embedded cryptosystems. We remark that the implementation of cryptographic algorithms could suffer attacks (SCA or FIA) for extracting the secret key. SCA (side-channel attacks) consist of passively recording some leakage to retrieve the key and FIA (fault injection attacks) consist of actively perturbing the computation to alter the output. One of the main sources of interest in LCD codes comes from the fact that they provide linear complementary pairs of codes.  A linear complementary pair of codes   $(C_1,C_2)$ consists of two codes in $\mathbb{F}_q^n$ with dimensions $k$ and $n-k$ such that $C_1+ C_2 = \mathbb{F}_q^n$. These pairs have been used in \cite{bri, CarletAttacks} for protecting implementations of symmetric cryptosystems against SCA, with level of protection depending on the minimum distance of $C_2$, and FIA, with level of protection depending on the minimum distance of $C_1$.

The above mentioned applications have produced a huge interest in LCD codes and many papers on this topic appeared very recently. Important contributions are \cite{Carlet2} and \cite{Pellikaan}, where the authors prove
that, for $q>3$, $q$-ary LCD codes are as good as $q$-ary linear codes. That is, for every linear code over a field $\mathbb{F}_q$ with more than 3 elements, one can construct an LCD code with the same parameters from that code. 

With respect to binary and ternary LCD codes, the best known LCD codes are reversible and are derived from BCH codes \cite{Jin, LiLi, rao, LiLiLiu}. As it is well-known, subfield-subcodes from codes over large fields can give rise to good codes over small fields. BCH codes are subfield-subcodes of Reed-Solomon codes and families of BCH LCD and cyclic LCD codes have been constructed in \cite{LiLi} and \cite{LiLiLiu} for few lengths. Some good binary reversible codes of odd length $n$, for $5 \leq n \leq 257$, are given in \cite{rao}, where the authors determine all the parameters for $5 \leq n \leq 99$.

In this paper we consider LCD codes coming from subfield-subcodes of the so-called $J$-affine variety codes. These codes are images of evaluation maps from vector spaces of polynomials in several variables generated by  suitable monomials. Our LCD codes may be regarded as a generalization of BCH codes, including extensions to the case of more variables, and allow us to reach a wider variety of lengths. Their metric structure and duality properties have been studied and successfully used to construct quantum stabilizer codes in previous works of the authors \cite{galindo-hernando, gal-her-rua, QINP,QINP2, ieee}.

Binary  $J$-affine variety codes for one variable with odd length  provide reversible codes which essentially coincide with those in \cite{rao}; however, our codes are derived from generic results (Theorems \ref{35} and \ref{te:ConDimension}) which can simplify some computations. Nevertheless, with one variable, we obtain unknown ternary reversible and binary  and ternary nonreversible LCD codes as shown in Section \ref{ejemplos}. The nonreversible codes obtained with only one variable have the same dimension and minimum distance as the reversible ones, but their length is one unit more.

Considering more than one variable, we get a much broader spectrum of lengths. Theorems \ref{pro:ejeSS} and \ref{hiperb}, and Remark \ref{bbb} provide a wide variety of new LCD codes with previously unknown lengths, having some of them good parameters. As a sample, in Section \ref{ejemplos} we give several families of LCD codes which, according to \cite{codetables}, contain many optimal or best known linear codes.

Decoding procedures may be useful for the cryptographic applications of LCD codes. Decoding algorithms have  been described for some families of codes considered in this paper \cite{fit, mar}. We believe that these algorithms may be adapted to all of them.

\section{LCD $J$-affine variety codes}
\label{secuno}

 In this section we consider $J$-affine variety codes. These linear codes were introduced in \cite{QINP} and used for constructing quantum codes. We review some results concerning self-orthogonality that will allow us to characterize LCD codes in this family. Finally, we give parameters for some families of LCD $J$-affine variety codes.

Along this paper, $q=p^r$ will be a positive power of a prime number $p$. Let $m \geq 1$ be an integer and fix $m$ integers $N_j>1$ such that $N_j-1$ divides $q-1$ for $j=1, 2, \ldots, m$. Let $\mathcal{R}:= \mathbb{F}_q [X_1,X_2, \ldots,X_m]$ be the ring of polynomials in $m$ variables and with coefficients in the finite field $\mathbb{F}_q $. Consider a subset $J \subseteq \{1,2, \ldots, m\}$ and the ideal $I_J$ in $\mathcal{R}$ generated by the binomials $X_j^{N_j} - X_j$ when $j \not \in J$ and by $X_j^{N_j -1} - 1$ otherwise. Set $Z_J = \{P_1, P_2, \ldots, P_{n_J}\}$ the zero-set of $I_J$ over $\mathbb{F}_q$. Note that the $j$th coordinate, for $j \in J$, of the points in $Z_J$ is different from zero and $n_J = \prod_{j \notin J} N_j \prod_{j \in J} (N_j -1)$. Furthermore, denote $T_j = N_j -2$  when $j \in J$ and $T_j = N_j -1$ otherwise; then define
$$
\mathcal{H}_J  = \{0,1,\ldots,T_1\}\times \{0,1,\ldots,T_2\} \times\cdots\times\{0,1,\ldots,T_m\}
$$
and, for any $\boldsymbol{a}=(a_1,a_2, \ldots, a_m) \in \mathcal{H}_J$, set $X^{\boldsymbol{a}} = X_1^{a_1} X_2^{a_2}\cdots X_{m}^{a_m}$.

Consider the quotient ring $\mathcal{R}_J:= \mathcal{R}/I_J$ and the evaluation map $\mathrm{ev}_J: \mathcal{R}_J \rightarrow \mathbb{F}_{q}^{n_J}$ given by $\mathrm{ev}_J(f) = \left(f(P_1), f(P_2), \ldots, f(P_{n_J}) \right)$, where $f$ denotes both the equivalence class and any polynomial representing it. As is well-known, ${\mathrm{ev}}_{{J}}$ is a bijection, and in particular $\{{\mathrm{ev}}_{J} (X^{\boldsymbol{a}}) \mid \boldsymbol{a} \in \mathcal{H}_J \}$ constitutes a basis for the image.

\begin{de}\label{def:unouno}
{\rm Let $\Delta$ be a non-empty subset of $\mathcal{H}_J $. The {\it $J$-affine variety code given by $\Delta$} is the $\mathbb{F}_q$-vector subspace $E^J_\Delta$ of $\mathbb{F}_q^{n_J}$ generated by $\mathrm{ev}_J (X^{\boldsymbol{a}})$, $\boldsymbol{a} \in \Delta$. We denote by $C^J_\Delta$ the (Euclidean) dual code of $E^J_\Delta$. }
\end{de}

Observe that the dimension of $E^J_\Delta$ equals the cardinality of $\Delta$, and consequently  the dimension of $C^J_\Delta$ is $n_J - \mathrm{card} (\Delta)$. Note that the univariate case contains the family of Reed-Solomon codes and for $J = \{1, 2, \ldots , m\}$ and $N_j =q$ for every $j$, one has a generalized toric code \cite{GTC}. It is also clear that the $J$-affine variety code $E^J_\Delta$ is LCD if and only if its dual code $C^J_\Delta$ is LCD.

The following result, which can be found in \cite[Proposition 1]{QINP}, gives the metric structure of $J$-affine variety codes.

\begin{pro}
\label{prop1}
Let $J \subseteq \{ 1 , 2, \ldots , m\}$. Consider $\boldsymbol{a}, \boldsymbol{b} \in \mathcal{H}_J$ and let $X^{\boldsymbol{a}}$ and $X^{\boldsymbol{b}}$ be two monomials representing elements in $\mathcal{R}_J$. Then, the inner product $\mathrm{ev}_J ( X^{\boldsymbol{a}}) \cdot \mathrm{ev}_J (X^{\boldsymbol{b}})$ is different from $0$ if,  and only if, the following two conditions are satisfied.
\begin{itemize}
\item For every $j \in J$, it holds that $a_j + b_j \equiv  0 \mod (N_j -1)$, (i.e.,  $a_j = N_j -1 - b_j$ when $a_j  + b_j > 0$ or $a_j=b_j=0$).
\item For every $j \notin J$, it holds that \begin{itemize}
\item either $a_j  + b_j > 0$ and $a_j + b_j \equiv 0 \mod (N_j -1)$,  (i.e.,  $a_j = N_j -1 - b_j$  if $0 < a_j, b_j < N_j -1$ or $(a_j,b_j) \in \left\{(0,N_j -1), (N_j -1,0), (N_j -1,N_j -1)  \right\}$ otherwise),
\item or $a_j = b_j = 0$ and $p \not | ~ N_j$.
\end{itemize}
\end{itemize}
\end{pro}

The following remark illustrates how to construct LCD $J$-affine variety codes. 

\begin{rem}
Fix $q= 3^3$, $m=2$, $J = \{ 1, 2\}$, $N = N_1=N_2=3^3$, and look for a set $\Delta \subset \mathcal{H}_J$ such that $E^J_\Delta$ is an LCD code. From Proposition \ref{prop1}, we deduce that the points in $\mathcal{H}_J$ can be divided into two sets. The first set consists of what we will call symmetric points, and they are $((N-1)/2,(N-1)/2)=(13,13)$, $(0,0)$,  $((N-1)/2,0)=(13,0)$ and  $(0,(N-1)/2)=(0,13)$. For a symmetric point $\boldsymbol{a}$, we have that $\mathrm{ev}_J (X^{\boldsymbol{a}})$ is orthogonal to $\mathrm{ev}_J (X^{\boldsymbol{b}})$ for all $\boldsymbol{b} \in \mathcal{H}_J \setminus \{\boldsymbol{a}$\} and $\mathrm{ev}_J (X^{\boldsymbol{a}}) \cdot \mathrm{ev}_J (X^{\boldsymbol{a}}) \neq 0$. Thus, suitable sets $\Delta$  can contain,  or not contain, symmetric points. The rest of the points in $\mathcal{H}_J$ are called asymmetric. In order to have an LCD code and when one desires $\Delta$ to contain an asymmetric point $(a_1,a_2)$, $a_1, a_2 \leq N-1$, the point $(N-1-a_1,N-1-a_2)$ (named reciprocal) must also be added to $\Delta$, and vice versa. Notice that, here, $N-1$ should be identified with zero. Indeed, $\mathrm{ev}_J (X^{(a_1,a_2)})$ is not orthogonal to $\mathrm{ev}_J (X^{(N-1-a_1,N-1-a_2)})$ and they are both orthogonal to $\mathrm{ev}_J (X^{\boldsymbol{b}})$ for every $\boldsymbol{b}$ different from $(a_1,a_2)$ and $(N-1-a_1,N-1-a_2)$. So to get suitable sets $\Delta$, we can consider any of the above given symmetric points and pairs as described, for instance one may have $(7,16), (19,10) \in \Delta$.

The procedure is a bit different when $J = \{2\}$ instead of $J = \{ 1, 2\}$. First we notice that in the case treated above the obtained dual code is also generated by the evaluation of monomials and, therefore,  it is a $J$-affine variety code. In this second case, assuming that we desire that $(0,10) \in \Delta$, our code be LCD and the dual code be also $J$-affine variety code, again by Proposition \ref{prop1}, we must add to $\Delta$ the points $(0,16),(26,16)$ and $(26,10)$.
\end{rem}

The following result formalizes the ideas in the previous remark characterizing LCD $J$-affine variety codes whose duals are again $J$-affine variety codes.

\begin{teo}\label{teo}
Let $\Delta$ be a subset of $\mathcal{H}_J$. The $J$-affine variety code $E_\Delta^J$ is LCD with its dual code also being $J$-affine variety if and only if $\Delta$ is a union of sets $\mathcal{R}_{\boldsymbol{a}}$ containing $\boldsymbol{a}$ and those elements $\boldsymbol{b} \in \mathcal{H}_J$ such that:
\begin{itemize}
\item For every $j \not\in J$, $b_j=N_j -1 -a_j $ if $0 < a_j < N_j -1$, and $b_j \in \{0, N_j -1\}$ otherwise.
\item For every $j \in J$, $b_j=N_j -1 -a_j $ if $0 < a_j < N_j -1$, and $b_j$ equals $0$  otherwise. Moreover $b_j$ may also be equal to $a_j$ in the case when either $a_i=0$ or $a_i=N_i-1$ for some $i \notin J$.
\end{itemize}

Any two distinct exponents $\boldsymbol{b}$ and $\boldsymbol{b}'$ in $\mathcal{R}_{\boldsymbol{a}}$ are called reciprocal, and $\boldsymbol{a}$ will be named symmetric whenever $\mathrm{card} (\mathcal{R}_{\boldsymbol{a}})=1$. Points that are not symmetric are called asymmetric.
\end{teo}

\begin{proof}
Let $\boldsymbol{a}$ and $\boldsymbol{b} \in \mathcal{R}_{\boldsymbol{a}}$ and assume $0 < a_j < N_j -1$ for $j \notin J$. By Proposition \ref{prop1}, $\mathrm{ev}_J (X^{\boldsymbol{a}})$ is not orthogonal to $\mathrm{ev}_J (X^{\boldsymbol{b}})$, and therefore $\boldsymbol{a}$, $\boldsymbol{b} \in \Delta$ to guarantee that $E_\Delta^J$ is LCD. It is also clear that if $\Delta = \mathcal{R}_{\boldsymbol{a}}$, the (Euclidean) dual code $C_\Delta^J$ is generated by the complement of $\Delta$ in $\mathcal{H}_J$. 

Finally, when $a_j = N_j -1$ or   $a_j = 0$ for $j \not\in J$, for constructing an LCD code whose dual is generated by monomials, one should have in $E_\Delta^J$, and not in $C_\Delta^J$, those vectors $\mathrm{ev}_J (X^{\boldsymbol{b}})$ which are not orthogonal to $\mathrm{ev}_J (X^{\boldsymbol{a}})$. This proves the result.
\end{proof}

\begin{rem}
\label{la15}
{\rm
The cardinality of  the sets $\mathcal{R}_{\boldsymbol{a}}$ described in Theorem \ref{teo} is a power of $2$. It is $1$ or $2$ if no coordinate of $\boldsymbol{a}$ equals $0$ or $N_j -1$ for some $j \not \in J$.

When $J \neq \{1,2, \ldots, m\}$ and $p$ does not divide $N_j$ for $j \not \in J$, one can also get LCD $J$-affine variety codes by including in $\Delta$ subsets $\mathcal{R'}_{\boldsymbol{a}}$ of $\mathcal{R}_{\boldsymbol{a}}$ with cardinality a power of $2$ whose elements have the $i$-th coordinate equal to either $0$ or $N_i -1$ for some indices $i$ in the set $\{1, 2, \ldots, m\} \setminus J$ and the corresponding evaluation vectors are not orthogonal. In this case, reasoning for $\Delta= \mathcal{R'}_{\boldsymbol{a}}$, the dual code is generated by the evaluation of the monomials in $\mathcal{H}_J \setminus \mathcal{R}_{\boldsymbol{a}}$ and polynomials which are linear combinations of monomials with exponents in $\mathcal{R}_{\boldsymbol{a}}$ and orthogonal to the evaluation of the monomials in $\mathcal{R'}_{\boldsymbol{a}}$. In generic cases, the dual space, contains a vector space with dimension $n_J - \mathrm{card}(\mathcal{R'}_{\boldsymbol{a}})$ which proves that $E_\Delta^J$ is an LCD code. {\it When considering this type of codes, we only consider the elements in $\mathcal{R'}_{\boldsymbol{a}}$ as reciprocal.}

As an easy example, setting $p=3$, $q=3^3$, $m=2$, $N_1=N_2=14$, $J= \{2\}$, $\boldsymbol{a}=(0,1)$ and $\Delta = \mathcal{R'}_{\boldsymbol{a}} = \{(0,1),(0,12)\}$, it holds that $E_\Delta^J$ is a LCD code of dimension $2$. Notice that $\Delta = \mathcal{R}_{\boldsymbol{a}} = \{(0,1),(0,12),(13,12),(13,1)\}$ gives another LCD code with dimension $4$.

}
\end{rem}

Some of the codes presented in \cite[Corollary 3.6]{Carlet1} can be recovered by considering the univariate case of $J$-affine variety codes, with $J=\{1\}$. The following result states parameters for LCD codes coming from the univariate case of $J$-affine variety codes. LCD codes obtained from subfield-subcodes of $J$-affine variety codes will be presented in the next section.

\begin{pro}\label{pro:uni}
Let $N$ be a positive integer such that $N-1$ divides $q-1$ and set another positive integer $1 \leq \delta \le (N-1)/2$ if $N-1$ is even and $\delta \le N/2 -1$ otherwise. For $J=\emptyset$ and $\Delta = \{ 0, 1, \ldots, \delta-1, N -\delta, \ldots , N-2, N-1  \}$, it holds that the dual code $C^J_\Delta$ of the $J$-affine variety code $E^J_\Delta$ is LCD with parameters $[N,N-2 \delta ,2\delta]_q$. Furthermore, for  $J=\{1\}$ and $\Delta = \{ 0, 1, \ldots, \delta-1, N -\delta, \ldots , N-2\}$,  the codes $E^J_\Delta$ and  $C^J_\Delta$ are LCD and MDS with parameters $[N-1, 2 \delta -1 , N - 2 \delta +1 ]_q$ and $[N-1,N-2 \delta ,2\delta]_q$, respectively.
\end{pro}

\begin{proof}
We prove the statement for the case when $J=\emptyset$. The proof is analogous when $J=\{1\}$. It is clear that $C^J_\Delta$ is the $J$-affine variety $E^J_{\Delta'}$ code given by $\Delta'= \{\delta, \delta+1, \ldots, \delta +(N-2 \delta -1)\}$. Now setting $\Delta''= \{1, 2, \ldots, N-2 \delta \}$, it holds that
$\{ \mathrm{ev}_J (X^{ {a}}) |  {a} \in \Delta'' \} = \{  \mathrm{ev}_J (X^{ {a}}) * \mathrm{ev}_J (X^{ N- \delta} ) : {a} \in \Delta'  \} $, where $*$ denotes the component-wise product. Since $\mathrm{wt}(\mathrm{ev}_J (X^{ N-\delta} )) = N-1$, both codes have the same parameters. So the dimension is clear and the distance follows from the fact that a polynomial of degree $N - 2\delta$ has at most $N - 2\delta$ zeroes.
\end{proof}

Now, for the general case and using Theorem \ref{teo}, we get a new family of LCD codes with a designed minimum distance. To prove it, we will need the following lemma which was stated in \cite[Proposition 4.1]{QINP2}.

\begin{lem}
\label{ddistancia}
Consider the ring $\mathcal{R}_J$ and fix a monomial ordering. Let $f(X_1,\ldots, X_m)$ be a polynomial of minimum total degree representing a class in $\mathcal{R}_J$ and let $X^{\boldsymbol{a}} = X_1^{a_1}   \cdots X_m^{a_m}$  be the leading monomial of $f$. Then
\[
\mathrm{card} \left\{ P \in Z_J \; | \; f(P) \neq 0 \right\} \geq \delta_{\boldsymbol{a}},
\]
where
\[
\delta_{\boldsymbol{a}} := \prod_{j=1}^m \left( N_j - \epsilon_j - a_j \right),
\]
$\epsilon_j$ being equal to $1$ if $j \in J$ and $\epsilon_j = 0$ otherwise.
\end{lem}

\begin{pro}\label{pro:eje}	
Keep the notation as at the beginning of this section setting $N_j>1$, $j=1, 2, \ldots,m$, such that $N_j-1$ divides $q-1$. Let $J=\{1,2, \ldots, m\}$ and fix $\alpha_j < T_j /2 $ if $T_j$ is even and $\alpha_j \le (T_j-1) /2 $ otherwise. Consider the subset of $\mathcal{H}_J$, $\Delta = L_1 \times L_2 \times \cdots \times L_m$ where $L_j = \{T_j/2 -\alpha_j, \ldots, T_j/2, \ldots, T_j/2 +\alpha_j \}$ if $T_j$ is even and $L_j = \{(T_j-1)/2 -\alpha_j, \ldots, (T_j-1)/2 +\alpha_j \}$ otherwise.

Then, writing $A_j = 2\alpha_j +1$, the code $C_\Delta^J$ is an LCD code with parameters
$$\left[n_J ,n_J- \prod_{j=1}^m A_j, \ge \min_{j\in J}\{ A_j+1\} \right]_q.$$
\end{pro}

\begin{proof}
Theorem \ref{teo} proves that $C_\Delta^J$ is LCD. Moreover, multiplying each generator of $E^J_\Delta$ by $\mathrm{ev}_J(1/\prod_{j \in J}X_j^{\beta_j})$ for suitable powers $\beta_j$, one obtains a monomially equivalent code (see \cite[Definition 3.1]{lit}) $E_{\Delta'}^J$ where the bottom left corner of the box $\Delta'$ is $\mathbf{0}$. The codes $E_{\Delta}^J$ and $E_{\Delta'}^J$ have the same dimension and distance and the same weight enumerators (see again \cite{lit}). Proposition \ref{prop1} shows that the dual code $C_{\Delta'}^J$
has the same minimum distance as the code $E_{\Delta^{''}}^J$, where
$$\Delta^{''} = \{0, \ldots, T_1 \} \times \{0, \ldots, T_2 \} \cdots \times \{0, \ldots, T_m  \}\setminus $$ $$\{0, T_1, T_1-1, \ldots,T_1- A_1 +1 \} \times \cdots \times \{0, T_m, T_m-1, \ldots,T_m- A_1 +1  \}.$$
Then, the result follows after applying Lemma \ref{ddistancia}. Notice that when $N_j -1= q -1$, $E_{\Delta^{''}}^J$  is a toric code and the result holds by \cite[Theorem 3]{Little} or \cite[Example 5.1]{Toric}.
\end{proof}

Now, and up to the end of this section, for providing a unified treatment according to the different sets $J$, we make a shift for the exponent of the monomials defining our code. Such a set is
$$\overline{\mathcal{H}}_J = \{\epsilon_1, \epsilon_1 +1, \ldots, \epsilon_1 + T_1 \} \times \{\epsilon_2, \epsilon_2 +1, \ldots, \epsilon_2 + T_2\} \times \cdots \times \{\epsilon_m, \epsilon_m +1, \ldots, \epsilon_m + T_m\}. $$
Identifying $T_j + \epsilon_j $ with $0$, for $j \in J$, we obtain a bijection from $\overline{\mathcal{H}}_J$ to $\mathcal{H}_J$. Note that $\overline{\mathcal{H}}_J$ and $\mathcal{H}_J$ are two different sets of exponents satisfying that the classes of the corresponding monomials in $\mathcal{R}_J$ are the same. Then, we consider the following set of monomials in $R$
\[
N(J,t) = \left\{  X^{\boldsymbol{b}} \; | \; \epsilon_j \leq b_j \leq N_j -1, \; 1 \leq j \leq m, \; \mathrm{and} \; \prod_{j=1}^m \left(b_j +1 - \epsilon_j \right) < t \right\},
\]
where $\epsilon_j =1$ if $j \in J$ and it equals zero otherwise. The hyperbolic code $\mathrm{Hyp}(J,t)$  \cite{SaintsHeegard,olav} can be defined as the (Euclidean) dual of the code given by the vector subspace of $\mathbb{F}_q^{n_J}$ generated by the evaluation by $\mathrm{ev}_J$  of the classes in $\mathcal{R}_J$ of the monomials in $N(J,t)$. By \cite[Proposition 4.3]{QINP2},
the minimum distance of $\mathrm{Hyp}(J,t)$ is larger than $t-1$. With the help of that code, we state
the following result which will be useful.

\begin{pro}
\label{hiper1}
With the notation as in the above paragraph and at the beginning of this section, set $N_j>1$, for $j=1, 2, \ldots,m$, such that $N_j-1$ divides $q-1$. Fix a positive integer such that $t \leq n_J =  \prod_{j \notin J} N_j \prod_{j \in J} (N_j -1)$, assume that $p|N_j$ for all $j \not \in J$ and consider the set $\Delta (J,t) = N(J,t) \cup N(J,t)^r$, where $N(J,t)^r$ is the set  of reciprocal elements (defined as in Theorem \ref{teo} or in Remark \ref{la15}) of those in $N(J,t)$, where we notice that for $j \in J$, $N_j -1$ must be identified with $0$. Then, the (Euclidean) dual $C_{\Delta (J,t)}^J$ of the $J$-affine variety code $E_{\Delta (J,t)}^J$ is a $J$-affine LCD code with parameters $[n_J, n_J- \mathrm{card} \left(\Delta (J,t)\right), \geq t]_q$.
\end{pro}

\begin{proof}
The construction of the code containing elements and reciprocal proves that we obtain an LCD code. The bound on the distance is also clear because we consider a code contained in the code $\mathrm{Hyp}(J,t)$ whose distance is larger than $t-1$.
\end{proof}

We are not directly interested in the LCD codes given by the above results because of
the recent papers \cite{Carlet2, Pellikaan} that show the existence of LCD codes for $q >3$ as good as linear codes. We will use them for obtaining suitable subfield-subcodes which will give rise to good binary and ternary LCD codes.

\section{LCD subfield-subcodes of $J$-affine variety codes}
\label{secdos}

Keep the notation as in Section \ref{secuno}. For $j \in J$, let $\mathbb{Z}_{\mathcal{T}_j} = \mathbb{Z}/\langle N_j -1  \rangle$ where we represent its classes by $\{0,1, \ldots, T_j\}$. For $j \not\in J$, we represent the classes of $\mathbb{Z}/\langle N_j -1  \rangle$ by $\{1, 2,  \ldots, T_j\}$ and define $\mathbb{Z}_{\mathcal{T}_j} = \{0\} \cup \mathbb{Z}/\langle N_j -1  \rangle $, where we represent its classes by $\{0,1, \ldots, T_j\}$. A subset $\mathfrak{I}$ of the Cartesian product $\mathbb{Z}_{\mathcal{T}_1}\times \mathbb{Z}_{\mathcal{T}_2} \times \cdots\times\mathbb{Z}_{\mathcal{T}_m}$ is called a {\it cyclotomic set}  with respect to $p$ if $p \cdot \boldsymbol{x} \in \mathfrak{I}$ for any $\boldsymbol{x} = (x_1, x_2, \ldots, x_m) \in \mathfrak{I}$, where $p \cdot \boldsymbol{x} = (p  x_1, p  x_2, \ldots, p x_m)$. $\mathfrak{I}$ is said to be {\it minimal} (with respect to $p$) whenever it contains all the elements that can be expressed as  $p^{ i } \cdot \boldsymbol{x}$ for some fixed element $\boldsymbol{x} \in \mathfrak{I}$ and some nonnegative integer $i$.  Within each minimal cyclotomic set $\mathfrak{I}$, we pick a representative $\boldsymbol{a} = (a_1, a_2, \ldots, a_m)$ given by nonnegative integers such that $a_1$ is the minimum of the first coordinates of the nonnegative representatives of the elements in $\mathfrak{I}$, $a_2$ is the minimum of the second coordinates of those elements in $\mathfrak{I}$ having $a_1$ as a first coordinate and the remaining coordinates, $a_3, \ldots, a_m$ are defined in the same way. We will denote by $\mathfrak{I}_{\boldsymbol{a}}$ the cyclotomic set $\mathfrak{I}$ with representative $\boldsymbol{a}$ and by $\mathcal{A}$ the set of representatives of the minimal cyclotomic sets. Thus, the set of minimal cyclotomic sets will be $\{ \mathfrak{I}_{\boldsymbol{a}}\}_{\boldsymbol{a} \in \mathcal{A}}$. In addition, we will denote $i_{\boldsymbol{a}} : = \mathrm{card}(\mathfrak{I}_{\boldsymbol{a}})$. Note that one can consider the cyclotomic sets with respect to an intermediate power $p^s$, such that $s$ divides $r$, however, since we only want to consider the case when $p$ equals 2 and 3, we set $s=1$.

Consider $\boldsymbol{a}$ and let $\boldsymbol{b}$ be a reciprocal of $\boldsymbol{a}$. Abusing the notation, let $\mathfrak{I}_{\boldsymbol{b}}$ be the cyclotomic set that contains $\boldsymbol{b}$. Taking into account the ring structure behind the two different sets $\mathbb{Z}_{\mathcal{T}_j}$, one gets the following straightforward result.

\begin{lem}\label{lem:coset}
Let $\boldsymbol{a} \in \mathcal{A}$ and let $\boldsymbol{b}$ be a reciprocal element. Then for every element in $\mathfrak{I}_{\boldsymbol{a}}$ there is a unique reciprocal element in $\mathfrak{I}_{\boldsymbol{b}}$ and both cyclotomic sets have the same cardinality. In addition, if $\boldsymbol{a}$ is asymmetric, then $\mathfrak{I}_{\boldsymbol{a}} \cap \mathfrak{I}_{\boldsymbol{b}} = \emptyset$.
\end{lem}

With the above notation, we say that a cyclotomic set $\mathfrak{I}_{\boldsymbol{a}}$ is {\it symmetric} if $\mathfrak{I}_{\boldsymbol{a}}=\mathfrak{I}_{\boldsymbol{b}}$ for all reciprocal element $\boldsymbol{b}$. Otherwise we will say that it is asymmetric. In addition, we define a partition of $\mathcal{A}$ as follows $\mathcal{A} =  \mathcal{A}_1 \cup \mathcal{A}_2$ ($\mathcal{A}_1 \cap \mathcal{A}_2 = \emptyset$), where $\mathcal{A}_1$ consists of the representatives of the symmetric cyclotomic sets and, for the asymmetric sets $\mathfrak{I}_{\boldsymbol{a}} \neq \mathfrak{I}_{\boldsymbol{a}'}$, where $\boldsymbol{a}$ and $\boldsymbol{a}'$ are reciprocal elements, we consider  $\boldsymbol{a}$ in $\mathcal{A}_1$ if  $\boldsymbol{a} <  \boldsymbol{a}'$ for the lexicographical ordering.

The {\it subfield-subcode} of a $J$-affine variety code $E_\Delta^J$ over $\mathbb{F}_{q}=\mathbb{F}_{p^r}$ is defined as $E_\Delta^{J,\sigma} := E_\Delta \cap \mathbb{F}_p^{n_J}$. Consider the following maps $\mathrm{tr}: \mathbb{F}_q \rightarrow \mathbb{F}_p$, $\mathrm{tr}(x) = x + x^p + \cdots + x^{p^{r-1}}$;  $\mathbf{tr}: \mathbb{F}_q^{n_J} \rightarrow \mathbb{F}_p^{n_J}$ given componentwise by $\mathrm{tr}(x)$, and $\mathcal{T}: R_J \rightarrow R_J$ defined by $\mathcal{T} (f) = f + f^p + \cdots + f^{p^{r-1}}$. We say that a class $f \in R_J$ evaluates to $\mathbb{F}_p$ whenever $f(\boldsymbol{a}) \in \mathbb{F}_p$ for all $\boldsymbol{a} \in Z_J$. In \cite[Proposition 5]{galindo-hernando} it is proved that $f$ evaluates to $\mathbb{F}_p$ if and only if $f= \mathcal{T}(g)$ for some $g \in R_J$. Now, considering for each $\boldsymbol{a} \in \mathcal{A}$, the close to $\mathcal{T}$ map,  $\mathcal{T}_{\boldsymbol{a} }: R_J \rightarrow R_J$, $\mathcal{T}_{\boldsymbol{a}} (f) = f + f^p + \cdots + f^{p^{(i_{\boldsymbol{a}} -1)}}$, we get the following result about the dimension of the code $E_\Delta^{J,\sigma} $. The proof is analogous to that in \cite[Theorem 4]{galindo-hernando}.

\begin{teo}
\label{ddimension}
Let $\Delta$ be a subset of $\mathcal{H}_J$ and set $\xi_{\boldsymbol{a}}$ a primitive element of the field $\mathbb{F}_{p^{i_{\boldsymbol{a}}}}$. Then, a basis of the vector space $E_\Delta^{J,\sigma}$ is given by the images under the map $\mathrm{ev}_J$ of the set of classes in $R_J$
$$\bigcup_{ \boldsymbol{a} \in \mathcal{A}| \mathfrak{I}_{\boldsymbol{a}} \subseteq \Delta
} \left\{ \mathcal{T}_{\boldsymbol{a}} (\xi_{\boldsymbol{a}}^{s} X^{\boldsymbol{a}}) | 0 \leq s \leq i_{\boldsymbol{a}} -1 \right\}.$$
\end{teo}

\subsection{Binary and ternary LCD subfield-subcodes coming from the univariate case}

We devote this section to provide binary and ternary LCD codes obtained as subfield-subcodes of univariate $J$-affine variety codes. The reasoning in Proposition \ref{pro:uni} and the above paragraphs in Section \ref{secdos} support the proof. We assume that $p$ equals 2 or 3.

\begin{pro}\label{pro:nok}
Let $N$ be a positive integer such that $N-1$ divides $q-1$. Recall that $q=p^r$ for a positive integer $r$. With the above notation, write $\mathcal{A}_1 = \{ a_0 =0 < a_1 < a_2 < \cdots < a_z \}$ the first set in the above given partition of $\mathcal{A}$. Let $t \in \{1, 2, \ldots, z\}$, and set $\Delta = \Delta_1 \cup \Delta_2$, where
$$ \Delta_1 = \mathfrak{I}_{a_0} \cup  \mathfrak{I}_{a_1} \cup \cdots \cup \mathfrak{I}_{a_t}$$ and
$\Delta_2$ is the union of the cyclotomic cosets with the reciprocal elements to those in $\Delta_1$. Then the dual code of  $E^{J,\sigma}_\Delta$ over $\mathbb{F}_p$ is LCD and has parameters: $[N-1, N-1 - \mathrm{card} (\Delta), \ge 2a_{t+1}  ]_p$ when $J=\{1\}$ and $[N, N - \mathrm{card} (\Delta), \ge 2a_{t+1}  ]_p$ otherwise ($J=\emptyset$).
\end{pro}

\begin{proof}
Theorem \ref{teo} and Lemma \ref{lem:coset} prove that our code is LCD. Theorem \ref{ddimension} determines the dimension of our code since the set $\Delta$ only contains complete cyclotomic sets. Finally, the same reasoning as in Proposition \ref{pro:uni} and the fact that we are considering subfields-subcodes proves the bound for the minimum distance. Notice that we have $2a_{t+1} -1$ consecutive elements in the dual case and that if $a_i$ is symmetric, the equality  $\mathfrak{I}_{a_i} = \mathfrak{I}_{a_{N-1-a_i}}$ holds.
\end{proof}

For the sake of generality, we provide formulae for the dimension under certain assumptions. First, we need some lemmas regarding cyclotomic sets, that are simply cyclotomic cosets for the one-variable case. The first one is essentially \cite[Lemma 8]{Aly}.

\begin{lem}\label{lem:aly}
Let $N>1$ be an integer such that $N-1$ divides $q-1$ and  assume that $p^{\lfloor r/2 \rfloor} < N-1 \le  p^r -1$. Then the cyclotomic sets $\mathfrak{I}_{a}$ have cardinality $r$  for all $1 \le a \le (N-1) p^{\lceil r/2\rceil}/(p^r -1)$.
\end{lem}

Next we characterize symmetric cyclotomic sets. Recall that $q=p^r$ and we are interested only in the cases $p=2$ and $p=3$.

\begin{lem}\label{lem:sere}
Let $N>1$ be an integer such that $N-1$ divides $q-1$, where $p \in \{2,3 \}$. Then, the cyclotomic set  $\mathfrak{I}_{a}$, with $a>0$, is symmetric  if and only if  $$a = \frac{N-1}{p^j +1},$$
for some $j \in \{0, 1,  \ldots, r-1 \}$ such that  $p^j +1$ is a divisor of $N-1$.
\end{lem}

\begin{proof}
It follows from the fact that $\mathfrak{I}_{a}$ is symmetric whenever there exists $j \in \{0, 1, \ldots, r-1 \}$ such that $a = N-1 - ap^j$, that is $a=(N-1)/(p^j +1)$.
\end{proof}

The following result gives sufficient conditions for asymmetry of cyclotomic sets when $m=1$.

\begin{pro}\label{pro:sufficient}
Keep the above notations, that is $N>1$ such that $N-1$ divides $q-1$ and $p \in \{2,3 \}$. Then:
\begin{itemize}
\item If $r$ is odd,  there are no symmetric cyclotomic set, unless when $p=3$ and $2$ divides $N-1$. In this case, the unique symmetric cyclotomic set is $\mathfrak{I}_{(N-1)/2}$.
\item Otherwise (for $r$ even),  $\mathfrak{I}_{a}$ is asymmetric if  $a < (N-1) /(p^\frac{r}{2} +1)$.
\end{itemize}

\end{pro}

\begin{proof}
For a start we consider the case when $r$ is odd. First we assume that $j=0$, then $p^j +1 =2$. When $p=2$, $q -1= 2^r -1 = 2(2^{r-1} -1) +1$ and so $N-1$ is odd,  therefore $p^j +1$ does not divide $N-1$ and there is no symmetric cyclotomic set by Lemma \ref{lem:sere}. In case $p=3$, if $N-1$ is even, then $p^j +1 $ divides $N-1$ and we have a cyclotomic symmetric set by Lemma \ref{lem:sere}.

Suppose now that $j>0$, write $r=kj + l$, $0 \leq l < j$, and consider the Euclidean division between the polynomials $X^r-1$ and $X^j +1$:
 $$
 X^r -1 = \left(X^{r-j} - X^{r-2j} + X^{r-3j} - \cdots + (-1)^{k-1} X^l \right) \left(X^j +1 \right) +  (-1)^{k} X^l -1.
 $$
Specializing $X$ to the value $p$, we get that if $j$ does not divide $r$ then  $p^j +1$ does not divide $q-1$. The same holds on the contrary, when $l=0$, since $r$ odd implies $k$ odd and the remainder is not zero.

Finally assume that $r$ is even. The symmetric cyclotomic set with smallest representative is given by the largest divisor of the form $p^j +1$ of $N-1$, for  $j \in \{0, 1, \ldots, r-1 \}$. The largest possible divisor is given by $j = r/2$, hence the representative of a symmetric set is larger than or equal to $(N-1) /(p^\frac{r}{2} +1)$ and the result holds.
\end{proof}

We are now ready to explicitly determine all the parameters of some of the codes described in Proposition \ref{pro:nok}. We consider the first cyclotomic set $\mathfrak{I}_{0}$, pairs of asymmetric cyclotomic sets and possibly, a symmetric cyclotomic and $\mathfrak{I}_{N-1}$. Actually, our next two results hold for any prime $p$.

\begin{teo}
\label{35}
Keep the above notation where $N$ is a positive integer such that $N-1$ divides $q-1$. Assume that $p^{\lfloor r/2 \rfloor} < N-1 \le  p^r -1$ and consider the first set of representatives of cyclotomic sets $\mathcal{A}_1 = \{ a_0 =0 < a_1 < a_2 < \cdots < a_z \}$ in the above given partition of $\mathcal{A}$. Let $t \in \{1, 2, \ldots, z\}$ be such that $a_t \le (N-1) p^{\lceil r/2\rceil}/(p^r -1)$ and set $\Delta = \Delta_1 \cup \Delta_2$, where
$$ \Delta_1 = \mathfrak{I}_{a_0} \cup  \mathfrak{I}_{a_1} \cup \cdots \cup \mathfrak{I}_{a_t}$$ and
$\Delta_2$ the union of the cyclotomic cosets with reciprocal elements to those in $\Delta_1$. Then,

\begin{itemize}
\item If $r$ is odd or if $r$ is even and $a_t \neq (N-1) p^{r/2}/(p^r -1)$, the dual code of  $E^{J,\sigma}_\Delta$, over $\mathbb{F}_p$, is LCD and has parameters: $[N-1, N- 2 t r , \ge 2a_{t+1}  ]_p$ when $J=\{1\}$ and $[N, N - 2 t r, \ge 2a_{t+1}  ]_p$ otherwise ($J=\emptyset$).
\item If $r$ is even and $a_t = (N-1) p^{r/2}/(p^r -1)$, the dual code of $E^{J,\sigma}_\Delta$,  over $\mathbb{F}_p$, is LCD and has parameters: $[N-1, N- (2 t-1) r , \ge 2a_{t+1}  ]_p$ when $J=\{1\}$ and $[N, N - (2 t-1) r, \ge 2a_{t+1}  ]_p$ otherwise ($J=\emptyset$).
\end{itemize}
\end{teo}

\begin{proof}
The bound for the minimum distance follows  from Proposition \ref{pro:nok}. Next we give a proof for the dimension of the codes.

If $r$ is odd or if $r$ is even and $a_t \neq (N-1) p^{r/2}/(p^r -1)$, then, by Lemma \ref{lem:aly}, the cardinality of all cyclotomic sets considered to define $\Delta$ is $r$, excepting $\mathfrak{I}_0$ (and occasionally $\mathfrak{I}_{N-1}$ if $J = \emptyset $); note that both sets have  cardinality 1. Moreover, by Proposition \ref{pro:sufficient}, the cyclotomic sets  $\mathfrak{I}_{a_j}$, $j \neq 0, N-1$, considered to define $\Delta$ are asymmetric, which concludes the proof.

If $r$ is even and $a_t = (N-1) p^{r/2}/(p^r -1)$, then the cardinality of all cyclotomic sets considered to define $\Delta$ (with the exception of $\mathfrak{I}_0$ and possibly $\mathfrak{I}_{N-1}$) is still $r$ by Lemma \ref{lem:aly}. Furthermore, by Proposition \ref{pro:sufficient}, all the cyclotomic sets considered to define $\Delta$ are asymmetric but $\mathfrak{I}_0$ and possibly $\mathfrak{I}_{N-1}$, and $\mathfrak{I}_{a_t}$ which is symmetric. Therefore, the equality $ 2 r (t-1) + r = (2t-1)r$ finishes the proof.
\end{proof}

To conclude this subsection, we prove that using Lemma 9 in \cite{Aly} one can avoid to consider representatives of cyclotomic sets, however in some cases, one will obtain codes with a smaller range of minimum distances. With our notation, Lemma 9 in \cite{Aly} is the following result.

\begin{lem}\label{le:TodosDistintos}
With the above notation, let $N$ be a positive integer such that $N-1 \mid p^r-1$ and suppose that $p^{\lfloor r/2 \rfloor} < N-1 \le  p^r -1$. If $x,y$ are distinct integers in the range  $1 \leq x,y \leq \min\{\lfloor (N-1)p^{\lceil r/2 \rceil}/(p^r-1)-1\rfloor,N-2\}$ which are not zero modulo $p$, then the cyclotomic cosets defined by $x$ and $y$ are different.
\end{lem}

The latter lemma determines an interval of integers  where the corresponding cyclotomic sets are all different and allows us to prove the following result.

\begin{teo}\label{te:ConDimension}
Let $q = p^r$, where $r$ is a positive integer and $p \in \{2,3\}$. Let $N$ be a positive integer such that $N-1$ divides $q-1$ and $p^{\lfloor r/2 \rfloor} < N-1 \le  p^r -1$. Then, for each integer $\delta$ such that
$2\le \delta\le
\min\{\lfloor (N-1)p^{\lceil r/2 \rceil}/(p^r-1)\rfloor,N-2\}$,
there exists an LCD code with length either  $N-1$ or $N$, designed minimum distance $\ge 2\delta$    and  dimension
$$
k=N-2(r\lceil (\delta-1)(1-1/p)\rceil).
$$
\end{teo}
\begin{proof}
We are considering sets $\Delta$ as above where $t$ is the largest integer such that $a_t < \delta \leq a_{t+1}$. Notice that the conditions in our statement also fulfil the conditions in Lemma \ref{lem:aly}, and therefore all the cyclotomic sets (with the exception of $\mathfrak{I}_0$ and possibly $\mathfrak{I}_{N-1}$) have cardinality $r$. Moreover, since $2\le \delta\le
\min\{\lfloor (N-1)p^{\lceil r/2 \rceil}/(p^r-1)\rfloor,N-2\}$, the representatives of the cyclotomic sets we use satisfy  $1 \leq a \leq \min\{\lfloor (N-1)p^{\lceil r/2 \rceil}/(p^r-1)-1\rfloor,N-2\}$. Under this condition,  Proposition \ref{pro:sufficient} states that we have no symmetric cyclotomic set (excepting $\mathfrak{I}_0$). Finally, Lemma \ref{le:TodosDistintos} warranties that, in order to compute the dimension of our codes, we only have to count how many integers, in the range of the statement, are not congruent with zero module $p$. The result holds since there are exactly $r\lceil (\delta-1)(1-1/p)\rceil$ such integers.
\end{proof}

\subsection{Binary and ternary LCD subfield-subcodes coming from the multivariate case}

In this section we state two results providing LCD codes which are not reversible codes. They are obtained as dual codes of subfield-subcodes of $J$-affine variety codes and reach lengths that are not achievable with BCH codes. Our first result considers subfield-subcodes of $J$-affine variety codes given by the union of cyclotomic sets whose representatives are in the box defined in Proposition \ref{pro:eje} and the second one is similar but taking representatives in the set $\Delta (J,t)$ defined in Proposition \ref{hiper1}. Using Lemma \ref{lem:coset}, they can be proved reasoning in a similar way as we did in Propositions \ref{pro:eje} and \ref{hiper1}. Our first result is the following.

\begin{teo}\label{pro:ejeSS}	
Let $N_j$, $1 \leq j \leq m$, be positive integers such that $N_j-1$ divides $q-1$. Assume that $J=\{1,2, \ldots, m\}$ and fix $\alpha_j < T_j /2 $ if $T_j$ is even and $\alpha_j \le (T_j-1) /2 $ otherwise. Consider the subset of $\mathcal{H}_J$, $\Delta = L_1 \times L_2 \times \cdots \times L_m$ where $L_j = \{T_j/2 -\alpha_j, \ldots, T_j/2, \ldots, T_j/2 +\alpha_j \}$ if $T_j$ is even and $L_j = \{(T_j-1)/2 -\alpha_j, \ldots, (T_j-1)/2 +\alpha_j \}$ otherwise. Consider the cyclotomic sets $\{\mathfrak{I}_{\boldsymbol{a}}\}_{\boldsymbol{a} \in \mathcal{A}}$ and let $\mathcal{A}_\Delta$ be  the set of representatives in $\mathcal{A}$ such that $\mathfrak{I}_{\boldsymbol{a}} \cap \Delta \neq \emptyset$. Set $\Delta^\sigma := \cup_{\boldsymbol{a} \in \mathcal{A}_\Delta} \mathfrak{I}_{\boldsymbol{a}}$.

Then, setting $A_j = 2\alpha_j +1$, the (Euclidean) dual code of the subfield-subcode $E^{J, \sigma}_{\Delta^\sigma}$ is an LCD code with parameters
$$\left[n_J ,n_J- \mathrm{card} (\Delta^\sigma), \ge \min_{j\in J}\{ A_j+1\} \right]_p.$$
\end{teo}

Finally we state the second result.

\begin{teo}
\label{hiperb}
Let $N_j$, $j=1, 2, \ldots,m$, be a positive integer such that $N_j-1$ divides $q-1$. Fix another positive integer $t$ such that $t \leq n_J =  \prod_{j \notin J} N_j \prod_{j \in J} (N_j -1)$, assume that $p|N_j$ for all $j \not \in J$ and consider the set $N(J,t)$ defined before Proposition \ref{hiper1}. Consider the cyclotomic sets $\{\mathfrak{I}_{\boldsymbol{a}}\}_{\boldsymbol{a} \in \mathcal{A}}$ and let $\mathcal{A}_{N(J,t)}$ be  the set of representatives in $\mathcal{A}$ such that $\mathfrak{I}_{\boldsymbol{a}} \cap N(J,t) \neq \emptyset$. Set $N(J,t)^\sigma := \bigcup_{\boldsymbol{a} \in \mathcal{A}_{N(J,t)}} \left(\mathfrak{I}_{\boldsymbol{a}} \cup \mathfrak{I}_{\boldsymbol{a}}^r\right)$, where $\mathfrak{I}_{\boldsymbol{a}}^r$ means the family of reciprocal  to $\mathfrak{I}_{\boldsymbol{a}}$ cyclotomic sets.

Then, the (Euclidean) dual of the subfield-subcode $E^{J, \sigma}_{N(J,t)^\sigma}$ is an LCD code with parameters $$[n_J, n_J- \mathrm{card} (N(J,t)^\sigma), \geq t]_p.$$

\end{teo}

\begin{rem}
\label{bbb}
{\rm
The construction in Theorem \ref{hiperb} can be improved from the point of view of subfield-subcodes when $J \neq \emptyset$ by noticing that the code Hyp$(J,t)^\perp$ is monomially equivalent  to $E^J_{N_0(J,t)}$ (see \cite{lit} for the definition and properties of monomially equivalent codes), where $N_0(J,t)$ is given by the monomials $X^{\boldsymbol{b}}/X^{\mathbf{\epsilon}}$, for  ${\boldsymbol{b}}$ in $N(J,t)$, where $X^{\mathbf{\epsilon}}$ is equal to $\prod_{j=1}^m X_j^{\epsilon_j}$ and $\epsilon_j$ as defined in Lemma \ref{ddistancia}. Then, with the same notation as in Theorem \ref{hiperb}, but replacing $N(J,t)$ with $N_0(J,t)$, we obtain LCD codes with parameters $[n_J, n_J- \mathrm{card} (N_0(J,t)^\sigma), \geq t]_p.$ Cyclotomic sets where some coordinates are zero have lower cardinality which improves the dimension of the dual codes. This approach will be used in some of our examples in the next section.
}
\end{rem}

\section{Examples}
\label{ejemplos}
The main references giving parameters of binary and ternary LCD codes are \cite{LiLi, LiLiLiu, rao}. All of them use BCH codes, the two first papers obtain LCD codes for concrete lengths and distances on arbitrary finite fields and the latter, from suitable representatives of cyclotomic cosets, computes parameters for some  binary LCD codes which, according to \cite{codetables}, are optimal or BKLC (best known linear codes). We will use this terminology along this section.

In the following two subsections, we will give examples of good binary and ternary LCD codes obtained with our results.

As regards binary LCD codes obtained from the univariate case, by using Theorems \ref{35} and \ref{te:ConDimension} we will be able to improve some codes in \cite{LiLi}  which are also given in \cite{rao}; in this particular case, the main advantages of our procedure are that we can avoid to compute cyclotomic sets (cosets in this case) and that we can get LCD codes with even length. Still on the univariate case but over $\mathbb{F}_3$,  we will provide new examples of ternary LCD codes which are optimal or  BKLC.

With respect to the multivariate case, Theorem \ref{hiperb} and especially its version in Remark \ref{bbb} give rise to generic families of binary and ternary LCD codes. Some of them are shown below and for some concrete values they provide new LCD codes which are optimal or BKLC.

\subsection{Binary LCD codes}
We devote this subsection to provide some examples of new binary LCD codes.

\begin{exa}\label{ex:uno} As a first example, let $p=2$ and $r=12$. By computing the cyclotomic sets, one can check that $\mathcal{A}_1=\{0,1,3,5,7,11,13\}$. Then according to Proposition \ref{pro:nok} we obtain codes with parameters:
$$[65,64, \geq 2]_2, [65,52,\geq 6]_2, [65,40, \geq 10]_2,[65,28,\geq 14]_2, [65,16,\geq 22]_2,[65,4,\geq 26]_2,$$ and
$$[66,64,\geq 2]_2, [66,52,\geq 6]_2, [66,40,\geq 10]_2,[66,28,\geq 14]_2, [66,16,\geq 22]_2,[66,4,\geq 26]_2.$$
Codes with odd length are obtained in \cite{rao}. Codes with even length are new and the codes with minimum distance $2$, $6$ and $10$ are BKLC.
\end{exa}

\begin{exa}\label{ex:dos} Theorem \ref{te:ConDimension} allows us to get new binary LCD codes with large length and minimum distance. For example, if we consider $p=2$, $r=14$ and $N=p^r$, we get LCD codes with minimum distances from $4$ to $256$. For instance, some parameters are $[16383, 14606, \geq 256]_2$, $[16383, 14620, \geq 254]_2$, $[16383, 14634, \geq 252]_2$.
\end{exa}

\begin{exa}\label{ex:tres} Now, we give an example of an optimal LCD code which can be obtained applying Remark \ref{bbb}. With the notation as in Theorem \ref{hiperb}, $p=2$, $r=4$, $J=\{1,2,3\}$, $N_1=16$ and $N_2=N_3=4$. Thus $n_J= 135$ and for $t=4$, it holds that
\[
N_0(J,t)^\sigma = \\\{ (0,0,0), (0,0,2), (0,0,1), (0,2,0), (0,1,0), (2,0,0), (4,0,0),
\]
\[
(8,0,0), (1,0,0), (14,0,0), (13,0,0), (11,0,0), (7,0,0)
\},
\]
and we obtain a code with parameters $[135,122,4]_2$ which is optimal.
\end{exa}

\begin{exa}\label{ex:cuatro} With the previous notation, consider $p=2$, $r=4$, $J=\{1,2, \ldots, m\}$, $N_1=N_2= \cdots= N_m= 4$ and $t=4$. Again by Remark \ref{bbb}, it holds that
\[
N_0(J,t)= \left\{(0,0, \ldots,0), (1,0, \ldots,0), (2,0, \ldots,0), \ldots, (0,0, \ldots,1), (0,0, \ldots,2) \right\}.
\]
Then, we get LCD codes with parameters $[3^m,3^m-2m-1, \geq 4]_2$. According to \cite{codetables}, these codes are optimal for $2 \leq m \leq 5$.

Another example with the same values $N_i$, $1 \leq i \leq m$, but larger minimum distance is obtained by setting $m=3$ and $t=12$. Then $N_0(J,t)^\sigma$ consists of the exponents of the monomials $X^{\boldsymbol{a}}/X^{\mathbf{\epsilon}}$ as defined in Remark \ref{bbb}, ${\boldsymbol{a}}$ in $\mathcal{H}_J$, excepting $$\{ (2,1,2), (1,2,1), (2,1,1), (1,2,2), (2,2,1),(1,1,2)\}.$$ Therefore, we get an LCD code with parameters $[27,6,12]_2$ which according to \cite{codetables} is optimal.  

Finally, writing $N_1=N_2= \cdots= N_m= 2$ and $J=\emptyset$, Theorem \ref{hiperb} for $t=4$ shows that we are considering as $N(J,4)^\sigma$ the elements in the axes and their reciprocal. So we obtain LCD codes with parameters $[2^m, 2^m - 2(m+1), \geq 4]_2$. For $m=7, 8$, the parameters are $[128,112, \geq 4]_2$ and $[256,238, \geq 4]_2$. These codes are not optimal, however according to \cite{codetables}, the obtained dimensions are only three units less than those of the corresponding optimal codes.
\end{exa}

\begin{exa}\label{ex:cinco} The same technique in Example \ref{ex:cuatro}, with $p=2$ and $r=4$, but decomposing $m=m_1+m_2$ and considering $N_1= N_2= \cdots =N_{m_1}=4$ and  $N_{m_1 +1}= N_{m_1 +2} =\cdots =N_{m}=6$ gives LCD codes with parameters $$[3^{m_1} 5^{m_2}, 3^{m_1} 5^{m_2} - 2 m_1- 4m_2-1, \geq 4]_2.$$

Some optimal LCD codes in this family have parameters $[45,36,4]_2$, $[75,64,4]_2$, $[81,72,4]_2$, $[125,112, 4]_2$ and $[200,187,4]_2$.

Analogously, one can consider $r=6$ and $N_1= N_2= \cdots =N_{m_1}=4$ and  $N_{m_1 +1}= N_{m_1 +2} =\cdots =N_{m}=8$,  obtaining LCD codes with parameters $[3^{m_1} 7^{m_2}, 3^{m_1} 7^{m_2} - 2m_1-6m_2-1, \geq 4]_2$. Within this family, there are optimal LCD codes with parameters $[63,5,4]_2$ and $[189,176,4]_2$.

To finish, we give two other families of binary LCD codes. Consider $N_1= 2^{k/2} +2$, $k$ even, and $N_2 = N_3 = \cdots = N_m = 4$. For suitable values of $r$, we get LCD codes with parameters $[3^{m-1} (2^{k/2}+1), 3^{m-1} (2^{k/2}+1) - 2 (m-1)-k-1, \geq 4]_2$ and $[3^{m-1} (2^{k/2}+1), 3^{m-1} (2^{k/2}+1) - 2 (m-1) -2k-1, \geq 6]_2$.  Some good LCD codes in these families have the following parameters: $[15,8,4]_2$, $[153,140,4]_2$, $[45,32,6]_2$, $[135,118,6]_2$, $[51,40, \geq 4]_2$ and $[153,132, \geq 6]_2$. All of them are optimal with the exception of the last two  which are BKLC.
\end{exa}

\subsection{Ternary LCD codes}
In this section we show some examples of good ternary LCD codes derived from our results.\\

\begin{exa}\label{ex:seis} In this example, we use again Proposition \ref{pro:nok} for giving new and good ternary LCD codes.  Set $p=3$, $r=4$ and $N=81$. Then $$\mathcal{A}_1=\{0,1,2,4,5,7,8,11,13,14,16,20,40\}$$ and computing the corresponding sets $\Delta$, we get LCD codes with parameters $[80, 63, \geq 8]_3$ and  $[80, 47, \geq 14]_3$ which are BKLC. Considering $r=5$ and $N=122$ one obtains an LCD code with parameters $[121, 100, \geq 8]_3$ which is also BKLC. Finally for $N=243$, LCD codes which are BKLC with parameters $[242, 221 \geq 8]_3$, $[242, 201 \geq 14]_3$, $[242, 181 \geq 20]_3$, $[242, 161 \geq 26]_3$ are obtained.

Consider now $p=3$ and $r=8$. After computing the corresponding cyclotomic sets, one can check that all of them (with the exception of $\mathfrak{I}_0$ in case $J=\emptyset$) are symmetric. Then $\mathcal{A}_1=\{0,1,2,4,5,7,8,11,13,14,16,41\}$.
Thus, we obtain codes with parameters:
$$[82,81,2]_3, [82,73,4]_3, [82,65,8]_3,[82,57,10]_3,[82,49,14]_3 ,[82,41,16]_3,
$$
$$
[82,33,22]_3,[82,25,26]_3,[82,17,28]_3,[82,9,32]_3,[82,1,82]_3.
$$
Moreover the codes with parameters
$$
[82,81,2]_3, [82,65,8]_3,[82,57,10]_3,[82,49,14]_3,[82,1,82]_3,
$$
are BKLC. Notice that in this last case we provide true minimum distance; the  three first codes were also obtained in \cite{LiLi}.

As mentioned, taking $J= \emptyset$, one obtains LCD codes with the same parameters except the length is one unit more.
\end{exa}

\begin{exa}\label{ex:siete}  With the same notation as in the above example, let $p=3$ and $r=6$. Setting $N=42$, it holds that $\mathcal{A}_1=\{0,1,2,4,7,8\}$ and we obtain ternary LCD codes with length $41$ and $42$ and the same dimension and minimum distance. The parameters in the first case are:
\[
[41,40,2]_3, [41,32,5]_3, [41,24,8]_3, [41,16,14]_3, [41,8,22]_3,
\]
where those with minimum distance $2$, $5$ and  $22$ are BKLC; as before, we are providing the true minimum distance.
\end{exa}

\begin{exa}\label{ex:ocho} Here we apply the same procedure we used for constructing the first family of LCD codes in Example \ref{ex:cinco}. Set $p=3$, $r=8$, $m=m_1 +m_2+m_3$, $N_1= N_2= \cdots =N_{m_1}=3$, $N_{m_1 +1}= N_{m_1= +2} =\cdots =N_{m_2}=5$ and $N_{m_1+ m_2 +1}= N_{m_1 +m_2 +2} =\cdots =N_{m}=6$. Then we get  LCD codes with parameters $$[2^{m_1} 4^{m_2} 5^{m_3}, 2^{m_1} 4^{m_2} 5^{m_3}- m_1-3m_2-4m_3-1, \geq 4]_3.$$
Some optimal codes in this family have the following parameters: $[16,11,4]_3$, $[32,26,4]_3$, $[128,120,4]_3$ and  $[64,57,4]_3$. A BKLC with parameters $[160,150, \geq 4]_3$ belongs also to the previous family.

An analogous reasoning as was given for the last family of codes in Example \ref{ex:cinco} gives rise to a new family of LCD codes with parameters $$[2^{m-1} (3^{k/2}+1), 2^{m-1} (3^{k/2}+1) - (m-1)-k-1, \geq 3]_3.$$
Some codes in this family have true minimum distance equal to $4$ with parameters  $[20,14,4]_3$, $[40,33,4]_3$, $[56,48,4]_3$ and $[164,154,4]_3$. The first two codes are optimal  and the last two are BKLC.

Finally, again for $p=3$, any $r$, $N_1=N_2 = \cdots = N_m=3$ and $J=\{2,3, \ldots, m\}$ we have that Remark \ref{bbb}, for $t=4$, gives a set $N_0(J,4)$ containing the elements of the axes and their reciprocal. When the non-vanishing coordinate is not the first coordinate, there is only one new reciprocal element and therefore we consider two elements in $N_0(J,4)^\sigma$; otherwise we must consider three elements instead, since one of them is symmetric.  This procedure gives rise to LCD codes with parameters $[3 \cdot 2^{m-1},  3 \cdot 2^{m-1}- 2m-1, \geq 4]_3$. For instance, for $m=7$, the parameters are $[192, 177, \geq 4]_3$; codes with the same parameters and distance one unit more are optimal.
\end{exa}


\begin{thebibliography}{10}

\bibitem{Aly}
S.A. Aly, A.~Klappenecker, and P.~K. Sarvepalli.
\newblock On quantum and classical {BCH} codes.
\newblock {\em IEEE Trans. Inform. Theory}, 53(3):1183--1188, 2007.

\bibitem{4P}
M.~Braun, T.~Etzion, and A.~Vardy.
\newblock Linearity and complements in projective space.
\newblock {\em Linear Algebra Appl.}, 430:57--70, 2013.

\bibitem{bri}
J.~Bringer~et al.
\newblock Orthogonal direct sum masking, a smartcard friendly computation
  paradigm in a code, with builtin protection against side-channel and fault
  attacks.
\newblock {\em Lect. Notes Comp. Sc.}, 8501:40--56, 2014.

\bibitem{CarletAttacks}
C.~Carlet and S.~Guilley.
\newblock Complementary dual codes for counter-measures to side-channel
  attacks.
\newblock {\em Adv. Math. Commun.}, 10(1):131--150, 2016.

\bibitem{Carlet1}
C.~Carlet, S.~Mesnager, C.~Tang, and Y.~Qi.
\newblock Euclidean and {H}ermitian {LCD} {MDS} codes.
\newblock {\em Arxiv, Preprint}, 1702.08033, 2017.

\bibitem{Carlet2}
C.~Carlet, S.~Mesnager, C.~Tang, and Y.~Qi.
\newblock Linear codes over $\mathbb{F}_q$ which are equivalent to {LCD} codes.
\newblock {\em Arxiv, Preprint}, 1703.04346, 2017.

\bibitem{fit}
J.~Fitzgerald and R.F. Lax.
\newblock Decoding affine variery codes using {G}r\"{o}bner basis.
\newblock {\em Des. Codes Crytogr.}, 13:147--158, 1998.

\bibitem{ieee}
C.~Galindo, O.~Geil, F.~Hernando, and D.~Ruano.
\newblock Improved constructions of nested code pairs.
\newblock {\em IEEE Trans. Inform. Theory}, IEEE Early Access Articles. DOI:
  10.1109/TIT.2017.2755682, 2017.

\bibitem{QINP2}
C.~Galindo, O.~Geil, F.~Hernando, and D.~Ruano.
\newblock On the distance of stabilizer quantum codes from {$J$}-affine variety
  codes.
\newblock {\em Quantum Inf. Process.}, 16(4):Art. 111, 32 pp, 2017.

\bibitem{galindo-hernando}
C.~Galindo and F.~Hernando.
\newblock Quantum codes from affine variety codes and their subfield subcodes.
\newblock {\em Des. Codes Crytogr.}, 76:89--100, 2015.

\bibitem{gal-her-rua}
C.~Galindo, F.~Hernando, and D.~Ruano.
\newblock New quantum codes from evaluation and matrix-product codes.
\newblock {\em Finite Fields Appl.}, 36:98--120, 2015.

\bibitem{QINP}
C.~Galindo, F.~Hernando, and D.~Ruano.
\newblock Stabilizer quantum codes from {$J$}-affine variety codes and a new
  {S}teane-like enlargement.
\newblock {\em Quantum Inf. Process.}, 14(9):3211--3231, 2015.

\bibitem{olav}
O.~Geil and T.~H{\o}holdt.
\newblock On hyperbolic codes.
\newblock {\em Lect. Notes Comp. Sc.}, 2227:159--171, 2001.

\bibitem{codetables}
M.~Grassl.
\newblock Bounds on the minimum distance of linear codes.
\newblock {\em www.codetables.de}, accessed on 30-09-2017.

\bibitem{20P}
X.~Hou and F.~Oggier.
\newblock On {LCD} codes and lattices.
\newblock {\em Proc. IEEE Int. Symp. on Inform. Theory}, pages 1501--1505,
  2016.

\bibitem{Jin}
L.~Jin.
\newblock Construction of {MDS} codes with complementary duals.
\newblock {\em IEEE Trans. Inform. Theory}, 63(5):2843--2847, 2017.

\bibitem{LiLi}
C.~Li, C.~Ding, and S.~Li.
\newblock {LCD} cyclic codes over finite fields.
\newblock {\em IEEE Trans. Inform. Theory}, 63(7):4344--4356, 2017.

\bibitem{LiLiLiu}
S.~Li, C.~Li, C.~Ding, and H.~Liu.
\newblock Two families of {LCD} {BCH} codes.
\newblock {\em IEEE Trans. Inform. Theory}, 63(9):5699--5717, 2017.

\bibitem{lit}
J.~Little and R.~Schwarz.
\newblock On $m$-dimensional toric codes.
\newblock {\em Arxiv, Preprint}, 0506102, 2005.

\bibitem{Little}
J.~Little and R.~Schwarz.
\newblock On toric codes and multivariate {V}andermonde matrices.
\newblock {\em Appl. Algebra Engrg. Comm. Comput.}, 18(4):349--367, 2007.

\bibitem{mar}
C.~Marcolla, E.~Orsino, and M.~Sala.
\newblock Improved decoding of affine variety codes.
\newblock {\em J. Pure Appl. Algebra}, 216:147--158, 2012.

\bibitem{mass}
J.L. Massey.
\newblock Reversible codes.
\newblock {\em Inf. Control}, 7:369--380, 1964.

\bibitem{Massey}
J.L. Massey.
\newblock Linear codes with complementary duals.
\newblock {\em Discrete Math.}, 106/107:337--342, 1992.
\newblock A collection of contributions in honour of Jack van Lint.

\bibitem{Pellikaan}
R.~Pellikaan.
\newblock {LCD} codes over $\mathbb{F}_q$ are as good as linear codes for $q$
  at least four.
\newblock {\em Arxiv, Preprint}, 1707.08856, 2017.

\bibitem{rao}
Y.~Rao~et al.
\newblock On binary {LCD} cyclic codes.
\newblock {\em Procedia Comp. Sc.}, 107:778--783, 2017.

\bibitem{Toric}
D.~Ruano.
\newblock On the parameters of {$r$}-dimensional toric codes.
\newblock {\em Finite Fields Appl.}, 13(4):962--976, 2007.

\bibitem{GTC}
D.~Ruano.
\newblock On the structure of generalized toric codes.
\newblock {\em J. Symbolic Comput.}, 44(5):499--506, 2009.

\bibitem{SaintsHeegard}
K.~Saints and C.~Heegard.
\newblock On hyperbolic cascaded {R}eed-{S}olomon codes.
\newblock In {\em Applied algebra, algebraic algorithms and error-correcting
  codes ({S}an {J}uan, {PR}, 1993)}, volume 673 of {\em Lecture Notes in
  Comput. Sci.}, pages 291--303. Springer, Berlin, 1993.

\bibitem{sen}
N.~Sendrier.
\newblock Linear codes with complementary duals meet the {G}ilbert-{V}arshamov
  bound.
\newblock {\em Discrete Math.}, 285:345--347, 2004.

\bibitem{t-h}
K.K. Tzeng and C.R.P. Hartmann.
\newblock On the minimum distance of certain reversible cyclic codes.
\newblock {\em IEEE Trans. Inform. Theory}, 16(5):644--646, 1970.

\bibitem{45P}
W.B. Vasantha~et al.
\newblock Erasure techniques in {MRD} codes.
\newblock {\em Zip publishing, Ohio}, 2012.

\bibitem{ya-ma}
X.~Yang and J.L. Massey.
\newblock The necessary and sufficient condition for a cyclic code to have a
  complementary dual.
\newblock {\em Discrete Math.}, 126:391--393, 1994.

\end{thebibliography}
\end{document}